\documentclass[english]{llncs}
\usepackage[T1]{fontenc}
\usepackage[latin9]{inputenc}
\usepackage{verbatim}
\usepackage{float}
\usepackage{amsmath}
\usepackage{amssymb}
\usepackage{esint}

\usepackage{geometry}
\makeatletter

\floatstyle{ruled}
\newfloat{algorithm}{tbp}{loa}
\providecommand{\algorithmname}{Algorithm}
\floatname{algorithm}{\protect\algorithmname}

\usepackage{color}

\ifdefined\DEBUG
\definecolor{marekgreen}{RGB}{0,185,0}

\newcommand{\fab}[1]{\textcolor{red}{#1}}

  \def\rem#1{{\marginpar{\raggedright\scriptsize #1}}}
  \newcommand{\fabr}[1]{\rem{\textcolor{red}{$\bullet$ #1}}}
  \newcommand{\joyr}[1]{\rem{\textcolor{blue}{$\bullet$ #1}}}
  \newcommand{\marr}[1]{\rem{\textcolor{marekgreen}{$\bullet$ #1}}}
\else

  \newcommand{\fab}[1]{#1}

  \newcommand{\fabr}[1]{}
  \newcommand{\joyr}[1]{}
  \newcommand{\marr}[1]{}

\fi

\usepackage{babel}

\makeatother

\usepackage{babel}
\begin{document}
\global\long\def\adj{\mbox{\footnotesize Adj}}
 \global\long\def\chr#1{\mathbf{1}_{#1}}
 \global\long\def\br#1{\left( #1 \right)}
 \global\long\def\brq#1{\left[ #1 \right]}
 \global\long\def\brw#1{\left\{  #1\right\}  }
 \global\long\def\cut#1{\partial#1 }
 \global\long\def\excond#1#2{\mathbb{E}\left[\left. #1 \right\vert #2 \right]}
 \global\long\def\ex#1{\mathbb{E}\left[#1\right]}
 \global\long\def\E{\mathbb{E}}
 \global\long\def\exls#1#2{\mathbb{E}_{#1}\left[#2\right]}
 \global\long\def\prcond#1#2{\mathbb{P}\left[\left. #1 \right\vert #2 \right]}
 \global\long\def\setst#1#2{\left\{  #1\left|#2\right.\right\}  }
 \global\long\def\setstcol#1#2{\left\{  #1:#2\right\}  }
 \global\long\def\set#1{\left\{  #1\right\}  }
 \global\long\def\adj#1{\delta\br{#1}}
 \global\long\def\setst#1#2{\left\{  \left.#1\right|#2\right\}  }
 \global\long\def\set#1{\left\{  #1\right\}  }
 \global\long\def\ind#1{\mathbf{1}\left[ #1 \right]}
 \global\long\def\st#1{[#1] }
 \global\long\def\opstyle#1{\mathbb{#1}}
 \global\long\def\size#1{\left|#1\right|}
 \global\long\def\setstcol#1#2{\left\{  #1:#2\right\}  }
 \global\long\def\set#1{\left\{  #1\right\}  }
 \global\long\def\indi#1{\chi\brq{#1}}
 \global\long\def\evalat#1#2{ #1 \Big|_{#2}}
 \global\long\def\prls#1#2{\opstyle P_{#1}\left[ #2 \right]}
 \global\long\def\pr#1{\opstyle P \left[ #1 \right]}
 \global\long\def\excondls#1#2#3{\mathbb{E}_{#1}\brq{\left.#2\right|#3}}
 \global\long\def\prcondls#1#2#3{\mathbb{P}_{#1}\brq{\left.#2\right|#3}}
 \global\long\def\st#1{[#1] }
 \global\long\def\indi#1{\chi\brq{#1}}
 \global\long\def\evalat#1#2{ #1 \Big|_{#2}}
 \global\long\def\Df#1#2{\frac{\partial#1}{\partial#2}}
 \global\long\def\hx#1{\hat{x}_{#1}}
 \global\long\def\eps{\varepsilon}
 \global\long\def\M{{\cal M}}
 \global\long\def\pr#1{\mathbb{P}\brq{#1}}
 \global\long\def\opstyle#1{\mathbb{#1}}
 \global\long\def\ex#1{\mathbb{E}\left[#1\right]}
 \global\long\def\xp#1{\mathbb{E}\left[#1\right]}
 \global\long\def\prcond#1#2{\opstyle P \left[\left. #1 \right\vert #2 \right]}
 \global\long\def\excond#1#2{\opstyle E \left[#1 \left|#2\right. \right]}
 \global\long\def\exls#1#2{\opstyle{\opstyle E}_{#1}\left[ #2 \right]}
 \global\long\def\prls#1#2{\opstyle P_{#1}\left[ #2 \right]}
 \global\long\def\br#1{\left( #1 \right)}
 \global\long\def\brq#1{\left[ #1 \right]}
 \global\long\def\brw#1{\left\{  #1\right\}  }
 \global\long\def\size#1{\left|#1\right|}
 \global\long\def\setst#1#2{\left\{  #1\left|#2\right.\right\}  }
 \global\long\def\setstcol#1#2{\left\{  #1:#2\right\}  }
 \global\long\def\pxE{\br{p_{e}x_{e}}_{e\in E}}
 \global\long\def\evalat#1#2{ #1 \Big|_{#2}}
 \global\long\def\X#1{\hat{X}_{#1}}
 \global\long\def\E{\hat{E}}
 \global\long\def\Frac#1#2{#1\left/\br{#2}\right.}
 \global\long\def\P#1{{\cal P}\br{#1}}
 \global\long\def\I{{\cal I}}
 \global\long\def\brqbb#1{\llbracket#1\rrbracket}
 \global\long\def\h#1{\hat{#1}}
 \global\long\def\lrg#1{#1_{large}}
 \global\long\def\sml#1{#1_{small}}

\pagestyle{headings}

\author{ Marek Adamczyk\inst{1} \and Fabrizio Grandoni\inst{2} \and Joydeep Mukherjee\inst{3} } \institute{ Department of Computer, Control, and Management Engineering, Sapienza University of Rome, Italy, \email{adamczyk@dis.uniroma1.it. }\and IDSIA, University of Lugano, Switzerland, \email{fabrizio@idsia.ch}. \and Institute of Mathematical Sciences, CIT, India, \email{joydeepm@imsc.res.in} } \title{Improved Approximation Algorithms\\ for Stochastic Matching \thanks{This work was partially done while the first and last authors were visiting IDSIA. The first and second authors were partially supported by the ERC StG project NEWNET no.~279352, and the first author by the ERC StG project PAAl no.~259515. The third author was partially supported by the ISJRP project Mathematical Programming in Parameterized Algorithms.}} \maketitle
\begin{abstract}
\noindent In this paper we consider the \emph{Stochastic Matching}
problem, which is motivated by applications in kidney exchange and
online dating. We are given an undirected graph in which every edge
is assigned a probability of existence and a positive profit, and each node is assigned a positive integer called \emph{timeout}. We know whether an edge exists or not only after probing it. On this random graph we are executing a process, which one-by-one
probes the edges and gradually constructs a matching. The process is constrained in two ways: once an edge is taken it cannot be removed
from the matching, and the timeout of node $v$ upper-bounds the number
of edges incident to $v$ that can be probed. The goal is to maximize
the expected profit of the constructed matching.

For this problem Bansal et al.~\cite{BGLMNR12algo} provided a $3$-approximation
algorithm for bipartite graphs, and a $4$-approximation for general
graphs. In this work we improve the approximation factors to $2.845$
and $3.709$, respectively. 

We also consider an online version of the bipartite case, where one
side of the partition arrives node by node, and each time a node $b$
arrives we have to decide which edges incident to $b$ we want to
probe, and in which order. Here we present a $4.07$-approximation,
improving on the $7.92$-approximation of Bansal et al.~\cite{BGLMNR12algo}.

The main technical ingredient in our result is a novel way of probing
edges according to a random but non-uniform permutation. Patching
this method with an algorithm that works best for large probability
edges (plus some additional ideas) leads to our improved approximation
factors. 
\end{abstract}

\section{Introduction}

In this paper we consider the \emph{Stochastic Matching} problem,
which is motivated by applications in kidney exchange and online dating.
Here we are given an undirected graph $G=(V,E)$. Each edge $e\in E$
is labeled with an (existence) probability $p_{e}\in(0,1]$ and a weight
(or profit) $w_{e}>0$, and each node $v\in V$ with a \emph{timeout}
(or \emph{patience}) $t_{v}\in\mathbb{N}^{+}$. An algorithm for this
problem probes edges in a possibly adaptive order. Each time an edge
is probed, it turns out to be \emph{present} with probability $p_{e}$,
in which case it is (irrevocably) included in the matching under construction
and provides a profit $w_{e}$. We can probe at most $t_{u}$ edges
among the set $\delta(u)$ of edges incident to node $u$ (independently
from whether those edges turn out to be present or absent). Furthermore,
when an edge $e$ is added to the matching, no edge $f\in\delta(e)$
(i.e., incident on $e$) can be probed in subsequent steps. Our goal
is to maximize the expected weight of the constructed matching. Bansal
et al.~\cite{BGLMNR12algo} provide an LP-based $3$-approximation
when $G$ is bipartite, and via reduction to the bipartite case a
$4$-approximation for general graphs (see also \cite{DBLP:conf/stacs/AdamczykSW14}).

We also consider the \emph{Online Stochastic Matching with Timeouts}
problem introduced in \cite{BGLMNR12algo}. Here we are given in input
a bipartite graph $G=(A\cup B,A\times B)$, where nodes in $B$ are
\emph{buyer types} and nodes in $A$ are \emph{items} that we wish
to sell. Like in the offline case, edges are labeled with probabilities
and profits, and nodes are assigned timeouts. However, in this case
timeouts on the item side are assumed to be unbounded. Then a second
bipartite graph is constructed in an online fashion. Initially this
graph consists of $A$ only. At each time step one random buyer $\tilde{b}$
of some type $b$ is sampled (possibly with repetitions) from a given
probability distribution. The edges between $\tilde{b}$ and $A$
are copies of the corresponding edges in $G$. The online algorithm
has to choose at most $t_{b}$ unmatched neighbors of $\tilde{b}$,
and probe those edges in some order until some edge $a\tilde{b}$
turns out to be present (in which case $a\tilde{b}$ is added to the
matching and we gain the corresponding profit) or all the mentioned
edges are probed. This process is repeated $n$ times, and our goal
is to maximize the final total expected profit\footnote{As in \cite{BGLMNR12algo}, we assume that the probability of a buyer
type $b$ is an integer multiple of $1/n$.}.

For this problem Bansal et al.~\cite{BGLMNR12algo} present a $7.92$-approximation
algorithm. In his Ph.D.~thesis Li~\cite{li2011decision} claims an
improved $4.008$-approximation. However, his analysis contains a
mistake \cite{Li14private}. By fixing that, he still achieves a
$5.16$-approximation ratio improving over~\cite{BGLMNR12algo}.

\subsection{Our Results}

Our main result is an approximation algorithm for bipartite Stochastic
Matching which improves the 3-approximation of Bansal et al.~\cite{BGLMNR12algo}
(see Section \ref{sec:offline}). \begin{theorem} \label{thr:mainOffline}
There is an expected $2.845$-approximation algorithm for Stochastic
Matching in bipartite graphs. \end{theorem}

Our algorithm for the bipartite case is similar to the one from~\cite{BGLMNR12algo},
which works as follows. After solving a proper LP and rounding the
solution via a rounding technique from~\cite{GKPS06}, Bansal et
al.~probe edges in uniform random order. Then they show that every
edge $e$ is probed with probability at least $x_e\cdot g(p_{max})$, where $x_{e}$ is the fractional value of $e$, 
$p_{max}:=\max_{f\in\delta(e)}\{p_{f}\}$
is the largest probability of any edge incident to $e$ ($e$ excluded), and $g(\cdot)$ is a
decreasing function with $g(1)=1/3$.

Our idea is to rather consider edges in a carefully chosen \emph{non-uniform}
random order. This way, we are able to show (with a slightly simpler analysis) that each edge $e$ is probed with probability $x_{e}\cdot g\br{p_{e}}\geq \frac{1}{3}x_e$. Observe that we have the same function $g(\cdot)$ as in \cite{BGLMNR12algo}, but depending on $p_e$ rather than $p_{max}$. In particular, according to our analysis, small probability edges are more likely to be probed than large
probability ones (for a given value of $x_{e}$), regardless of the
probabilities of edges incident to $e$. Though this approach alone
does not directly imply an improved approximation factor, it is not
hard to patch it with a simple greedy algorithm that behaves best
for large probability edges, and this yields an improved approximation
ratio altogether.

We also improve on the $4$-approximation for general graphs in~\cite{BGLMNR12algo}.
This is achieved by reducing the general case to the bipartite one
as in prior work, but we also use a refined LP with blossom inequalities
in order to fully exploit our large/small probability patching technique.
\begin{theorem} \label{thr:mainOfflineGeneral} There is an expected
$3.709$-approximation algorithm for Stochastic Matching in general
graphs. \end{theorem}

Similar arguments can also be successfully applied to the online case.
By applying our idea of non-uniform permutation of edges we would get
a $5.16$-approximation (the same as in~\cite{li2011decision}, after
correcting the mentioned mistake). However, due to the way edges have
to be probed in the online case, we are able to finely control the
probability that an edge is probed via \emph{dumping factors}. This
allows us to improve the approximation from $5.16$ to $4.16$. Our
idea is similar in spirit to the one used by Ma~\cite{DBLP:conf/soda/Ma14}
in his neat 2-approximation algorithm for correlated non-preemptive
stochastic knapsack. Further application of the large/small probability
trick gives an extra improvement down to $4.07$ (see Section \ref{sec:online}).
\begin{theorem} \label{thr:mainOnline} There is an expected $4.07$-approximation
algorithm for Online Stochastic Matching with Timeouts. \end{theorem}

\subsection{Related work}

The Stochastic Matching problem falls under the framework of adaptive
stochastic problems presented first by Dean et al.~\cite{DBLP:journals/mor/DeanGV08}.
Here the solution is in fact a process, and the optimal one might
even require larger than polynomial space to be described.

The Stochastic Matching problem was originally presented by Chen et
al.~\cite{CIKMR09} together with applications in kidney exchange
and online dating. The authors consider the unweighted version of
the problem, and prove that a greedy algorithm is a $4$-approximation.
Adamczyk~\cite{A11ipl} later proved that the same algorithm is in
fact a $2$-approximation, and this result is tight. The greedy algorithm
does not provide a good approximation in the weighted case, and all
known algorithms for this case are LP-based. Here, Bansal et al.~\cite{BGLMNR12algo}
showed a 3-approximation for the bipartite case. Adamczyk~\cite{nonnegativesubmodularprobing}
presented a different analysis of the same algorithm. Via a reduction
to the bipartite case, Bansal et al.~\cite{BGLMNR12algo} also obtain
a 4-approximation algorithm for general graphs. The same approximation
factor is obtained by Adamczyk et al.~\cite{DBLP:conf/stacs/AdamczykSW14}
using iterative randomized rounding.

\section{Stochastic Matching}

\label{sec:offline}

\subsection{Bipartite graphs\label{sub:Bipartite-graphs}}

Let us denote by $OPT$ the optimum probing strategy, and let $\ex{OPT}$
denote its expected outcome. Consider the following LP: 
\begin{align}
\max & \sum_{e\in E}w_{e}p_{e}x_{e} & \br{\mbox{LP-BIP}}\nonumber \\
\mbox{s.t.} & \sum_{e\in\delta(u)}p_{e}x_{e}\leq1, & \forall u\in V;\label{cardinalityConstraints}\\
 & \sum_{e\in\delta(u)}x_{e}\leq t_{u}, & \forall u\in V;\label{toleranceConstraints}\\
 & 0\leq x_{e}\leq1, & \forall e\in E.
\end{align}
The proof of the following Lemma is already quite standard~\cite{DBLP:conf/stacs/AdamczykSW14,BGLMNR12algo,DBLP:journals/mor/DeanGV08}
--- just note that $x_{e}=\pr{OPT\mbox{ probes }e}$ is a feasible
solution of LP-BIP. 
\begin{lemma}\label{lem:Bansal} \cite{BGLMNR12algo} Let $LP_{bip}$
be the optimal value of \emph{LP-BIP}. It holds that $LP_{bip}\geq\ex{OPT}$.
\end{lemma}

Our approach is similar to the one of Bansal et al.~\cite{BGLMNR12algo}
(see also Algorithm \ref{alg:bipartite} in the figure). We solve
LP-BIP: let $x=(x_e)_{e\in E}$ be the optimal fractional solution. Then we apply
to $x$ the rounding procedure by Gandhi et al.~\cite{GKPS06}, which
we shall call just GKPS. Let $\hat{E}$ be the set of rounded edges,
and let $\hat{x}_{e}=1$ if $e\in\hat{E}$ and $\hat{x}_{e}=0$ otherwise.
GKPS guarantees the following properties of the rounded solution: 
\begin{enumerate}
\item (Marginal distribution) For any $e\in E$, $\pr{\hat{x}_{e}=1}=x_{e}.$ 
\item (Degree preservation) For any $v\in V$, $\sum_{e\in\delta(v)}\hat{x}_{e}\leq\lceil\sum_{e\in\delta(v)}x_{e}\rceil\leq t_{v}.$ 
\item (Negative correlation) For any $v\in V$, any subset $S\subseteq\delta(v)$
of edges incident to $v$, and any $b\in\{0,1\}$, it holds that $\pr{\wedge_{e\in S}(\hat{x}_{e}=b)}\leq\prod_{e\in S}\pr{\hat{x}_{e}=b}.$ 
\end{enumerate}
Our algorithm sorts the edges in $\hat{E}$ according to a random
permutation and probes each edge $e\in\hat{E}$ according to that
order, but provided that the endpoints of $e$ are not matched already.
It is important to notice that, by the degree preservation property, in
$\hat{E}$ there are at most $t_{v}$ edges incident to each node
$v$. Hence, the timeout constraint of $v$ is respected even if the
algorithm probes all the edges in $\delta(u)\cap\hat{E}.$

Our algorithm differs from~\cite{BGLMNR12algo} and subsequent work
in the way edges are randomly ordered. Prior work exploits a random
uniform order on $\hat{E}$. We rather use the following, more complex
strategy. For each $e\in \hat{E}$ we draw a random variable $Y_{e}$ distributed
on the interval $\brq{0,\frac{1}{p_{e}}\ln\frac{1}{1-p_{e}}}$ according
to the following cumulative distribution: $\pr{Y_{e}\leq y}=\frac{1}{p_{e}}\br{1-e^{-p_{e}y}}.$
Observe that %
the density function of $Y_{e}$ in this interval is $e^{-yp_{e}}$
(and zero otherwise). Edges of $\hat{E}$ are sorted in increasing order of the $Y_{e}$'s,
and they are probed according to that order. We next
let $Y=(Y_{e})_{e\in \hat{E}}$.

Define $\hat{\delta}(v):=\adj v\cap\E$. We say that an edge $e\in\hat{E}$
is \emph{safe} if, at the time we consider $e$ for probing, no other
edge $f\in\hat{\delta}(e)$ is already taken into the matching. Note
that the algorithm can probe $e$ only in that case, and if we do
probe $e$, it is added to the matching with probability $p_{e}$.

\begin{algorithm}
\label{alg:bipartite} \protect\protect\caption{Approximation algorithm for bipartite Stochastic Matching.}

\begin{enumerate}
\item Let $\br{x_{e}}_{e\in E}$ be the solution to LP-BIP. 
\item Round the solution $\br{x_{e}}_{e\in E}$ with GKPS; let $(\h x_{e})_{e\in E}$
be the rounded 0-1 solution, and $\h E=\{e\in E|\h x_{e}=1\}$. 
\item For every $e\in\h E$, sample a random variable $Y_{e}$ distributed
as $\pr{Y_{e}\leq y}=\frac{1-e^{-yp_{e}}}{p_{e}}$. 
\item For every $e\in\h E$ in increasing order of $Y_{e}$:

\begin{enumerate}
\item If no edge $f\in\hat{\delta}(e):=\delta(e)\cap\hat{E}$ is yet taken,
then probe edge $e$ 
\end{enumerate}
\end{enumerate}
\label{alg:bipartite} 
\end{algorithm}

The main ingredient of our analysis is the following lower-bound on
the probability that an arbitrary edge $e$ is safe.\begin{lemma}
\label{lem:safe-bip} For every edge $e$ it holds that $\prcond{e\mbox{ is safe}}{e\in\hat{E}}\geq g\br{p_{e}}$,
where $$g\br p:=\frac{1}{2+p}\br{1-\exp\br{-\br{2+p}\frac{1}{p}\ln\frac{1}{1-p}}}.$$\end{lemma}

\begin{proof} In the worst case every edge $f\in\hat{\delta}(e)$
that is before $e$ in the ordering can be probed, and each of these
probes has to fail for $e$ to be safe. Thus 
\[
\prcond{e\mbox{ is safe}}{e\in\E}\geq\excondls{\E\setminus e,Y}{\prod_{f\in\hat{\delta}(e):Y_{f}<Y_{e}}\br{1-p_{f}}}{e\in\E}.
\]
Now we take expectation on $Y$ only, and using the fact that the
variables $Y_{f}$ are independent, we can write the latter expectation
as 
\begin{align}
 & \excondls{\E\setminus e}{\int_{0}^{\frac{1}{p_{e}}\ln\frac{1}{1-p_{e}}}\br{\prod_{f\in\hat{\delta}(e)}\br{\pr{Y_{f}\leq y}(1-p_{f})+\pr{Y_{f}>y}}}e^{-p_{e}\cdot y}\mbox{d}y}{e\in\E}.\label{eq:conditional}
\end{align}
Observe that $\pr{Y_{f}\leq y}\br{1-p_{f}}+\pr{Y_{f}>y}=1-p_{f}\pr{Y_{f}\leq y}.$
When $y>\frac{1}{p_{f}}\ln\frac{1}{1-p_{f}}$, then $\pr{Y_{f}\leq y}=1$,
and moreover, $\frac{1}{p_{f}}(1-e^{-p_{f}\cdot y})$ is an increasing
function of $y$. Thus we can upper-bound $\pr{Y_{f}\leq y}$ by $\frac{1}{p_{f}}(1-e^{-p_{f}\cdot y})$
for any $y\in\brq{0,\infty}$, and obtain that $1-p_{f}\pr{Y_{f}\leq y}\geq1-p_{f}\frac{1}{p_{f}}(1-e^{-p_{f}\cdot y})=e^{-p_{f}\cdot y}.$
Thus~\eqref{eq:conditional} can be lower bounded by 
\begin{align*}
 & \excondls{\E\setminus e}{\int_{0}^{\frac{1}{p_{e}}\ln\frac{1}{1-p_{e}}}e^{-\sum_{f\in\hat{\delta}\br e}p_{f}\cdot y-p_{e}\cdot y}\mbox{d}y}{e\in\E}\\
= & \excondls{\E\setminus e}{\frac{1}{\sum_{f\in\hat{\delta}\br e}p_{f}+p_{e}}\br{1-e^{-\br{\sum_{f\in\hat{\delta}\br e}p_{f}+p_{e}}\frac{1}{p_{e}}\ln\frac{1}{1-p_{e}}}}}{e\in\E}.
\end{align*}

\end{proof}

From the negative correlation and marginal distribution properties
we know that $\excondls{\E\setminus e}{\h x_{f}}{e\in\E}\leq\exls{\h E\setminus e}{\h x_{f}}=x_{f}$
for every $f\in\delta\br e$, and therefore $\excondls{\E\setminus e}{\sum_{f\in\hat{\delta}\br e}p_{f}}{e\in\h E}\leq\sum_{f\in\delta\br e}p_{f}x_{f}\leq2$,
where the last inequality follows from the LP constraints. Consider
function $f(x):=\frac{1}{x+p_{e}}\br{1-e^{-\br{x+p_{e}}\frac{1}{p_{e}}\ln\frac{1}{1-p_{e}}}}$.
This function is decreasing and convex. From Jensen's inequality we
know that $\ex{f(\fab{x})}\geq f(\ex{\fab{x}})$. Thus 
\begin{multline*}
\excondls{\E\setminus e}{f\br{\sum_{f\in\hat{\delta}\br e}p_{f}}}{e\in\h E}\geq f\br{\excondls{\E\setminus e}{\sum_{f\in\hat{\delta}\br e}p_{f}}{e\in\h E}}\\
\hfill\geq f(2)=\frac{1}{2+p_{e}}\br{1-e^{-\br{2+p_{e}}\frac{1}{p_{e}}\ln\frac{1}{1-p_{e}}}}=g(p_{e}).\hfill\square
\end{multline*}

From Lemma \ref{lem:safe-bip} and the marginal distribution property,
the expected contribution of edge $e$ to the profit of the solution
is 
\[
w_{e}p_{e}\cdot\pr{e\in\hat{E}}\cdot\prcond{e\text{ is safe}}{e\in\hat{E}}\geq w_{e}p_{e}x_{e}\cdot g(p_{e})\geq w_{e}p_{e}x_{e}\cdot g(1)=\frac{1}{3}w_{e}p_{e}x_{e}.
\]
Therefore, our analysis implies a $3$ approximation, matching the
result in \cite{BGLMNR12algo}. However, by playing with the probabilities
appropriately we can do better.

\paragraph{Patching with Greedy.} We next describe an improved approximation algorithm, based on the
patching of the above algorithm with a simple greedy one. Let $\delta\in(0,1)$
be a parameter to be fixed later. We define $\lrg E$ as the (\emph{large})
edges with $p_{e}\geq\delta$, and let $E_{small}$ be the remaining
(\emph{small}) edges. Recall that $LP_{bip}$ denotes the optimal
value of LP-BIP. Let also $\lrg{LP}$ and $\sml{LP}$ be the fraction
of $LP_{bip}$ due to large and small edges, respectively; i.e., $\lrg{LP}=\sum_{e\in\lrg E}w_{e}p_{e}x_{e}$
and $\sml{LP}=LP_{bip}-\lrg{LP}$. Define $\gamma\in[0,1]$ such that
$\gamma LP_{bip}=LP_{large}$. By refining the above analysis, we
obtain the following result. \begin{lemma} \label{lem:bipRefined}
Algorithm \ref{alg:bipartite} has expected approximation ratio $\frac{1}{3}\gamma+g(\delta)\br{1-\gamma}$.
\end{lemma}

\begin{proof} The expected profit of the algorithm is at least: 
\begin{multline*}
\sum_{e\in E}w_{e}p_{e}x_{e}\cdot g(p_{e})\geq\sum_{e\in E_{large}}w_{e}p_{e}x_{e}\cdot g(1)+\sum_{e\in E_{small}}w_{e}p_{e}x_{e}\cdot g(\delta)\\
=\frac{1}{3}LP_{large}+g(\delta)LP_{small}=\left(\frac{1}{3}\gamma+g(\delta)\br{1-\gamma}\right)LP_{bip}.\hfill\square
\end{multline*}
\end{proof}

Consider the following greedy algorithm. Compute a maximum weight
matching $M_{grd}$ in $G$ with respect to edge weights $w_{e}p_{e}$,
and probe the edges of $M_{grd}$ in any order. Note that the timeout
constraints are satisfied since we probe at most one edge incident
to each node (and timeouts are strictly positive by definition and
w.l.o.g.). \begin{lemma} \label{lem:greedy} The greedy algorithm
has expected approximation ratio $\delta\gamma$. \end{lemma}

\begin{proof} It is sufficient to show that the expected profit of
the obtained solution is at least $\delta\cdot\lrg{LP}$. Let $x=(x_{e})_{e\in E}$
be the optimal solution to LP-BIP. Consider the solution $x'=(x'_{e})_{e\in E}$
that is obtained from $x$ by setting to zero all the variables corresponding
to edges in $E_{small}$, and by multiplying all the remaining variables
by $\delta$. Since $p_{e}\geq\delta$ for all $e\in\lrg E$, $x'$
is a feasible fractional solution to the following matching LP: 
\begin{align}
\max & \sum_{e\in E}w_{e}p_{e}z_{e} & \text{(LP-MATCH)}\nonumber \\
\mbox{s.t.} & \sum_{e\in\delta(u)}z_{e}\leq1, & \forall u\in V;\nonumber \\
 & 0\leq z_{e}\leq1, & \forall e\in E.
\end{align}
The value of $x'$ in the above LP is $\delta\cdot LP_{large}$ by
construction. Let $LP_{match}$ be the optimal profit of LP-MATCH.
Then $LP_{match}\geq\delta\cdot LP_{large}$. Given that the graph
is bipartite, LP-MATCH defines the matching polyhedron, and we can
find an integral optimal solution to it. But such a solution is exactly
a maximum weight matching according to weights $w_{e}p_{e}$, i.e. $\sum_{e\in M_{grd}}w_{e}p_{e}=LP_{match}$. The claim follows since the expected profit of the greedy algorithm is precisely the weight of $M_{grd}$. $\hfill\square$ 
\end{proof}

The overall algorithm, for a given $\delta$, simply computes the
value of $\gamma$, and runs the greedy algorithm if $\gamma\delta\geq\left(\frac{1}{3}\gamma+g(\delta)\br{1-\gamma}\right)$,
and Algorithm \ref{alg:bipartite} otherwise\footnote{Note that we cannot run both algorithms, and take the best solution.}.

The approximation factor is given by $\max\{\frac{\gamma}{3}+(1-\gamma)g(\delta),\gamma\delta\}$,
and the worst case is achieved when the two quantities are equal,
i.e., for $\gamma=\frac{g\br{\delta}}{\delta+g\br{\delta}-\frac{1}{3}}$,
yielding an approximation ratio of $\frac{\delta\cdot g\br{\delta}}{\delta+g\br{\delta}-\frac{1}{3}}$.
Maximizing (numerically) the latter function in $\delta$ gives $\delta=0.6022$,
and the final $2.845$-approximation ratio claimed in Theorem~\ref{thr:mainOffline}.

\subsection{General graphs}

\label{sec:arbitrarygraphs}

For general graphs, we consider the linear program LP-GEN which is
obtained from LP-BIP by adding the following \emph{blossom inequalities}:
\begin{align}
 & \sum_{e\in E(W)}p_{e}x_{e}\leq\frac{|W|-1}{2} & \forall W\subseteq V,|W|\text{ odd}.\label{blossomConstraints}
\end{align}
Here $E(W)$ is the subset of edges with both endpoints in $W$. We remark that, using standard tools from matching theory, we can solve LP-GEN in polynomial time despite its exponential number of constraints; see the book of Schrijver for details~\cite{Schrijver:book}. Also in this case $x_{e}=\pr{OPT\mbox{ probes }e}$ is a feasible
solution of LP-GEN, hence the analogue of Lemma \ref{lem:Bansal} still holds.

Our Stochastic Matching algorithm for the case of a general graph $G=\br{V,E}$
works via a reduction to the bipartite case. First we solve LP-GEN;
let $x=\br{x_{e}}_{e\in E}$ be the optimal fractional solution. Second we randomly split
the nodes $V$ into two sets $A$ and $B$, with $E_{AB}$ being
the set of edges between them. On the bipartite graph $\br{A\cup B,E_{AB}}$
we apply the algorithm for the bipartite case, but using the fractional
solution $\br{x_{e}}_{e\in E_{AB}}$ induced by LP-GEN rather than
solving LP-BIP. Note that $\br{x_{e}}_{e\in E_{AB}}$ is a feasible
solution to LP-BIP for the bipartite graph $\br{A\cup B,E_{AB}}$.

The analysis differs only in two points w.r.t.~the one for the bipartite
case. First, with $\h E_{AB}$ being the subset of edges of $E_{AB}$
that were rounded to 1, we have now that $\pr{e\in\h E_{AB}}=\pr{e\in E_{AB}}\cdot\prcond{e\in\h E_{AB}}{e\in E_{AB}}=\frac{1}{2}x_{e}$.
Second, but for the same reason, using again the negative correlation and marginal distribution properties, we have
\begin{align*}
 & \excondls{}{\sum_{f\in\h{\delta}\br e}p_{f}}{e\in\h E_{AB}}\leq \sum_{f\in\delta\br e}p_{f}\pr{f\in\h E_{AB}}=\frac{1}{2}\sum_{f\in\delta\br e}p_{f}x_{f}\leq\frac{1}{2}(2-2p_{e}x_{e})\leq1.
\end{align*}
Repeating the steps of the proof of Lemma~\ref{lem:safe-bip} and
including the above inequality we get the following. \begin{lemma}
\label{lem:safe-gen} For every edge $e$ it holds that $\prcond{e\mbox{ is safe}}{e\in\hat{E}_{AB}}\geq h\br{p_{e}}$,
where $$h\br{p}:=\frac{1}{1+p}\br{1-\exp\br{-\br{1+p}\frac{1}{p}\ln\frac{1}{1-p}}}.$$
\end{lemma}

Since $h(p_{e})\geq h(1)=\frac{1}{2}$, we directly obtain a $4$-approximation
which matches the result in \cite{BGLMNR12algo}. Similarly to the
bipartite case, we can patch this result with the simple greedy algorithm
(which is exactly the same in the general graph case). For a given
parameter $\delta\in[0,1]$, let us define $\gamma$ analogously to
the bipartite case. Similarly to the proof of Lemma \ref{lem:bipRefined},
one obtains that the above algorithm has approximation factor $\frac{\gamma}{4}+\frac{1-\gamma}{2}h(\delta)$.
Similarly to the proof of Lemma \ref{lem:greedy}, the greedy algorithm
has approximation ratio $\gamma\delta$ (here we exploit the blossom
inequalities that guarantee the integrality of the matching polyhedron).
We can conclude similarly that in the worst case $\gamma=\frac{h\br{\delta}}{2\delta+h\br{\delta}-1/2}$,
yielding an approximation ratio of $\frac{\delta\cdot h\br{\delta}}{2\delta+h\br{\delta}-1/2}$.
Maximizing (numerically) this function over $\delta$ gives, for $\delta=0.5580$,
the $3.709$ approximation ratio claimed in Theorem \ref{thr:mainOfflineGeneral}.

\section{Online Stochastic Matching with Timeouts}

\label{sec:online}

Let $G=\br{A\cup B,A\times B}$ be the input graph, with items $A$
and buyer types $B$. We use the same notation for edge probabilities,
edge profits, and timeouts as in Stochastic Matching. 
Following \cite{BGLMNR12algo}, we can assume w.l.o.g. that each buyer
type is sampled uniformly with probability $1/n$. Consider the following
linear program:

\begin{align*}
\max & \sum_{a\in A,b\in B}w_{ab}p_{ab}x_{ab} & \mbox{(LP-ONL)}\\
\mbox{s.t.} & \sum_{b\in B}p_{ab}x_{ab}\leq1, & \forall a\in A\\
 & \sum_{a\in A}p_{ab}x_{ab}\leq1, & \forall b\in B\\
 & \sum_{a\in A}x_{ab}\leq t_{b}, & \forall b\in B\\
 & 0\leq x_{ab}\leq1, & \forall ab\in E.
\end{align*}
The above LP models a bipartite Stochastic Matching instance
where one side of the bipartition contains exactly one buyer per buyer
type. In contrast, in the online case several buyers of the same buyer
type (or none at all) can arrive, and the optimal strategy can allow
many buyers of the same type to probe edges. 
Still, that is not a problem since the following lemma from \cite{BGLMNR12algo}
allows us just to look at the graph of buyer types and not at the
actual realized buyers. \begin{lemma} \label{lem:onlineBansal} (\cite{BGLMNR12algo},
Lemmas 9 and 11) Let $\ex{OPT}$ be the expected profit of the optimal
online algorithm for the problem. Let $LP_{onl}$ be the optimal value
of \emph{LP-ONL. }It holds that $\ex{OPT}\leq LP_{onl}$. \end{lemma}

\global\long\def\Ab{A_{b}}

We will devise an algorithm whose expected outcome is at least $\frac{1}{4.07}\cdot LP_{onl}$,
and then Theorem \ref{thr:mainOnline} follows from Lemma \ref{lem:onlineBansal}.

\paragraph{The algorithm.}

We initially solve LP-ONL and let $\br{x_{ab}}_{ab\in A\times B}$
be the optimal fractional solution. Then buyers arrive. When a buyer
of type $b$ is sampled, then 1) if a buyer of the same type $b$
was already sampled before we simply discard her, do nothing, and
wait for another buyer to arrive, 2) if it is the first buyer of type
$b$, then we execute the following \emph{subroutine for buyers.}
Since we take action only when the first buyer of type $b$ comes,
we shall denote such a buyer simply by $b$, as it will not cause
any confusion.

\paragraph{Subroutine for buyers.}

Let us consider the step of the online algorithm in which the first
buyer of type $b$ arrived, if any. Let $\Ab$ be the items that are still
available when $b$ arrives. Our subroutine will probe a subset of
at most $t_{b}$ edges $ab$, $a\in\Ab$. Consider the vector $\br{x_{ab}}_{a\in\Ab}$.
Observe that it satisfies the constraints $\sum_{a\in\Ab}p_{ab}x_{ab}\leq1$
and $\sum_{a\in\Ab}x_{ab}\leq t_{b}$. Again using GKPS, we round
this vector in order to get $\br{\hat{x}_{ab}}_{a\in\Ab}$ with $\hat{x}_{ab}\in\{0,1\}$,
and satisfying the marginal distribution, degree preservation, and
negative correlation properties\footnote{Actually in this case we have a bipartite graph where one side has
only one vertex, and here GKPS reduces to Srinivasan's
rounding procedure for level-sets \cite{DBLP:conf/focs/Srinivasan01}.}. 
Let $\hat{A}_{b}$ be the set of items $a$ such that $\hat{x}_{ab}=1$.
For each $ab$, $a\in\h A_{b}$, we independently draw a random variable
$Y_{ab}$ with distribution: $\pr{Y_{ab}<y}=\frac{1}{p_{ab}}\br{1-\exp\br{-p_{ab}\cdot y}}$
for $y\in\brq{0,\frac{1}{p_{ab}}\ln\frac{1}{1-p_{ab}}}$. Let $Y=\br{Y_{ab}}_{a\in\h A_{b}}$.

Next we consider items of $\hat{A}_{b}$ in increasing order of $Y_{ab}$.
Let $\alpha_{ab}\in[\frac{1}{2},1]$ be a \emph{dumping factor} that
we will define later. With probability $\alpha_{ab}$ we probe edge
$ab$ and as usual we stop the process (of probing edges incident
to $b$) if $ab$ is present. Otherwise (with probability $1-\alpha_{ab}$)
we \emph{simulate} the probe of $ab$, meaning that with probability
$p_{ab}$ we stop the process anyway --- like if edge $ab$ were probed
and turned out to be present. Note that we do not get any profit from
the latter simulation since we do not really probe $ab$.

\paragraph{Dumping factors.}

It remains to define the dumping factors. For a given edge $ab$,
let 
\[
\beta_{ab}:=\excondls{\h A{}_{b}\setminus a,Y}{\prod_{a'\in A_{b}:Y_{a'b}<Y_{ab}}\br{1-p_{a'b}}}{a\in\h A_{b}}.
\]
Using the inequality $\sum_{a\in\Ab}p_{ab}x_{ab}\leq1$, by repeating the analysis from Section~\ref{sec:offline} we can
show that 
\[
\beta_{ab}\geq h(p_{ab})=\frac{1}{1+p_{ab}}\br{1-\exp\br{-\br{1+p_{ab}}\frac{1}{p_{ab}}\ln\frac{1}{1-p_{ab}}}}\geq\frac{1}{2}.
\]
Let us assume for the sake of simplicity that we are able to compute
$\beta_{ab}$ exactly. We will show in Appendix \ref{sec:computeDumping} how to remove
this assumption. We set $\alpha_{ab}=\frac{1}{2\beta_{ab}}$. Note
that $\alpha_{ab}$ is well defined since $\beta_{ab}\in[1/2,1]$.

\paragraph{Analysis.}

Let us denote by ${\cal A}_{b}$ the event that at least one buyer
of type $b$ arrives. The probability that an edge $ab$ is probed
can be expressed as: 
\[
\pr{{\cal A}_{b}}\cdot\prcond{\mbox{no }b'\mbox{ takes }a\mbox{ before }b}{{\cal A}_{b}}\cdot\prcond{b\mbox{ probes }a}{{\cal A}_{b}\wedge a\mbox{ is not yet taken}}.
\]
The probability that $b$ arrives is $\pr{{\cal A}_{b}}=1-\br{1-\frac{1}{n}}^{n}\geq1-\frac{1}{e}$.
We shall show first that $$\prcond{b\mbox{ probes }a}{{\cal A}_{b}\wedge a\mbox{ is not yet taken}}$$
is exactly $\frac{1}{2}x_{ab}$, and later we shall show that $\prcond{\mbox{no }b'\mbox{ takes }a\mbox{ before }b}{{\cal A}_{b}}$
is at least $\frac{1}{1+\frac{1}{2}\br{1-\frac{1}{e}}}$. This will
yield that the probability that $ab$ is probed is at least 
\[
\br{1-\frac{1}{e}}\frac{1}{1+\frac{1}{2}\br{1-\frac{1}{e}}}\cdot\frac{1}{2}x_{ab}=\frac{e-1}{3e-1}x_{ab}>\frac{1}{4.16}x_{ab}.
\]

Consider the probability that some edge $a'b$ appearing before $ab$
in the random order \emph{blocks} edge $ab$, meaning that $ab$ is
not probed because of $a'b$. Observe that each such $a'b$ is indeed
considered for probing in the online model, and the probability that
$a'b$ blocks $ab$ is therefore $\alpha_{a'b}p_{a'b}+(1-\alpha_{a'b})p_{a'b}=p_{a'b}$.
We can conclude that the probability that $ab$ is not blocked is
exactly $\beta_{ab}$.

Due to the dumping factor $\alpha_{ab}$, the probability that we
actually probe edge $ab\in\hat{A}_{b}$ is exactly $\alpha_{ab}\cdot\beta_{ab}=\frac{1}{2}$.
Recall that $\pr{a\in\hat{A}_{b}}=x_{ab}$ by the marginal distribution
property. Altogether 
\begin{equation}
\prcond{b\mbox{ probes }a}{{\cal A}_{b}\wedge a\mbox{ is not yet taken}}=\frac{1}{2}x_{ab}.\label{eq:afterdumping}
\end{equation}

Next let us condition on the event that buyer $b$ arrived, and let us lower bound the probability that $ab$ is not blocked on the $a$'s side in
such a step, i.e., that no other buyer has taken $a$ already. The
buyers, who are first occurrences of their type, arrive uniformly
at random. Therefore, we can analyze the process of their arrivals
as if it was constructed by the following procedure: every buyer $b'$
is given an independent random variable $Y_{b'}$ distributed exponentially
on $[0,\infty]$, i.e., $\pr{Y_{b'}<y}=1-e^{y}$; buyers arrive in
increasing order of their variables $Y_{b'}$. Once buyer $b'$ arrives,
it probes edge $ab'$ with probability (exactly) $\alpha_{ab'}\beta_{ab'}x_{ab'}=\frac{1}{2}x_{ab'}$
--- these probabilities are independent among different buyers. Thus,
conditioning on the fact that $b$ arrives, we obtain the following
expression for the probability that $a$ is safe at the moment when
$b$ arrives: 
\begin{eqnarray*}
 &  & \prcond{\mbox{no }b'\mbox{ takes }a\mbox{ before }b}{{\cal A}_{b}}\\
 & \geq & \excondls{}{\prod_{b'\in B\setminus b:Y_{b'}<Y_{b}}\br{1-\prcond{{\cal A}_{b'}}{{\cal A}_{b}}\prcond{b'\mbox{ probes }ab'}{{\cal A}_{b'}}p_{ab'}}}{{\cal A}_{b}}\\
 & = & \int_{0}^{\infty}\prod_{b'\in B\setminus b}\br{1-\prcond{{\cal A}_{b'}}{{\cal A}_{b}}\cdot\prcond{Y_{b'}<y}{{\cal A}_{b'}}\cdot\prcond{b'\mbox{ probes }ab'}{{\cal A}_{b'}}p_{ab'}}e^{-y}\mbox{d}y.
\end{eqnarray*}
Now let us upper-bound each of the probability factors in the above
product. First of all $\prcond{{\cal A}_{b'}}{{\cal A}_{b}}=1-\br{1-\frac{1}{n}}^{n-1}\leq1-\frac{1}{e}$.
Second, $\prcond{Y_{b'}<y}{{\cal A}_{b'}}=1-e^{-y}$ just by definition\footnote{The ${\cal A}_{b'}$ event in the condition simply indicates that
$Y_{b'}$ was drawn.}. Third, from~\eqref{eq:afterdumping} we have that $\prcond{b'\mbox{ probes }ab'}{{\cal A}_{b'}}=\frac{x_{ab}}{2}.$

Thus the above integral can be lower bounded by 
\begin{eqnarray*}
 &  & \int_{0}^{\infty}\prod_{b'\in B\setminus b}\br{1-\br{1-\frac{1}{e}}\br{1-e^{-y}}\cdot \frac{1}{2}x_{ab'}\cdot p_{ab'}}e^{-y}\mbox{d}y\\
 & \geq & \int_{0}^{\infty}\prod_{b'\in B\setminus b}\exp\br{-\br{1-\frac{1}{e}}\frac{1}{2}x_{ab'}\cdot p_{ab'}\cdot y}e^{-y}\mbox{d}y\\
 & = & \frac{1}{1+\br{1-\frac{1}{e}}\frac{1}{2}\br{\sum_{b'\in B\setminus b}p_{ab'}\cdot x_{ab'}}}\\
 & \geq & \frac{1}{1+\frac{1}{2}\br{1-\frac{1}{e}}}=\frac{2e}{3e-1}.
\end{eqnarray*}
Above in the first inequality we used the fact that $1-c(1-e^{-y})\geq e^{-cy}$
for $c\in[0,1]$ and any $y\in\mathbb{R}$: here $c=\br{1-\frac{1}{e}}\frac{1}{2}x_{ab'}\cdot p_{ab'}$.
In the first equality we used $\int_{0}^{\infty}e^{-ax}\mbox{d}x=\frac{1}{a}$.
In the last inequality we used the LP constraint $\sum_{b'\in B\setminus b}p_{ab'}\cdot x_{ab'}\leq1$.

Altogether, as anticipated earlier, 
\[
\pr{ab\mbox{ is probed}}\geq\br{1-\frac{1}{e}}\frac{x_{ab}}{2}\cdot\frac{2e}{3e-1}=x_{ab}\cdot\frac{e-1}{3e-1}>\frac{1}{4.16}\cdot x_{ab}.
\]

In Appendix~\ref{sec:computeDumping} we will show how to compute
the dumping factors so that the above probability is $\frac{e-1}{3e-1}+\eps$
for an arbitrarily small constant $\eps>0$. In particular, by choosing
a small enough $\eps$ the factor $4.16$ is still guaranteed.

We can again use the approach with big and small probabilities, thus reducing the approximation factor to $4.07$. The details are given in Appendix~\ref{sec:onlineBigSmall}. Theorem \ref{thr:mainOnline} follows.

 \bibliographystyle{plain}
\bibliography{StochasticMatching}

\appendix

\section{Combination with Greedy in the Online Case}

\label{sec:onlineBigSmall}


Recall that $h(p)=\frac{1}{1+p}\br{1-\exp\br{-\br{1+p}\frac{1}{p}\ln\frac{1}{1-p}}}$.
We are again applying the big/small probabilities trick, so let
$\delta\in(0,1)$ be a parameter to be fixed later. Consider back again the subroutine for buyers.
Previously we have used dumping factors $\alpha_{ab}=\frac{1}{2\beta_{ab}}$,
where --- recall --- $\beta_{ab}\geq h\br{p_{ab}}$.

This time we define $\alpha_{ab}=\frac{1}{\beta_{ab}}h\br{\delta}$
for $ab$ such that $p_{ab}\leq\delta$, and $\alpha_{ab}=\frac{1}{\beta_{ab}}\frac{1}{2}$
otherwise. We again assume here that we can calculate $\beta_{ab}$
(see Appendix~\ref{sec:computeDumping}). Define $E_{large}=\setst{ab\in E}{p_{ab}\geq\delta}$
and $E_{small}=E\setminus E_{large}$, and let $LP_{large}=\gamma\cdot LP_{onl}$.
Therefore, for edge $ab$ the probability that $ab$ is probed when
$b$ scans items is exactly $h\br{\delta}$ for $ab\in E_{small}$
and $\frac{1}{2}$ for $ab\in E_{large}.$ Now by repeating the steps
in the proof of Section \ref{sec:online}, we obtain that the probability that $ab$ is
not blocked on $a$'s side is at least 
\begin{align*}
\frac{1}{1+\br{1-\frac{1}{e}}\br{\sum_{b'\in B\setminus b}p_{ab'}\cdot\alpha_{ab'}\beta_{ab'}\cdot x_{ab'}}}\geq & \frac{1}{1+\br{1-\frac{1}{e}}h\br{\delta}\br{\sum_{b'\in B\setminus b}p_{ab'}\cdot x_{ab'}}}\\
\geq & \frac{1}{1+\br{1-\frac{1}{e}}h\br{\delta}},
\end{align*}
since $\alpha_{ab'}\cdot\beta_{ab'}=h\br{\delta}$ for small
edges and $\alpha_{ab'}\cdot\beta_{ab'}=\frac{1}{2}\leq h\br{\delta}$
for large edges. Therefore, the approximation ratio of such an algorithm
is at least 
\begin{multline*}
\br{1-\frac{1}{e}}\br{\gamma\frac{1/2}{1+h\br{\delta}\br{1-\frac{1}{e}}}+\br{1-\gamma}\frac{h(\delta)}{1+h\br{\delta}\br{1-\frac{1}{e}}}}\\
=\br{1-\frac{1}{e}}\frac{1}{1+h\br{\delta}\br{1-\frac{1}{e}}}\br{\gamma\frac{1}{2}+\br{1-\gamma}h\br{\delta}}.
\end{multline*}

An alternative algorithm simply computes a maximum weight matching
w.r.t. weights $p_{e}w_{e}$ in the graph corresponding to LP-ONL,
and upon arrival of the first copy of a buyer type $b$ probes only the edge incident to
$b$ in the matching (if any). By the same argument as in the offline
case, this matching has weight at least $\gamma\cdot\delta\cdot LP_{onl}$,
and every buyer type is sampled with probability at least $1-\frac{1}{e}$.
So the approximation ratio of the greedy algorithm is at least $\br{1-\frac{1}{e}}\gamma\delta$.


For a fixed $\delta$, depending on the value of $\gamma$ (that
we can compute offline) we can run the algorithm with best approximation
ratio according to the above analysis. Thus the overall approximation
ratio is 
\[
(1-\frac{1}{e})\max\left\{ \frac{1}{1+h\br{\delta}\br{1-\frac{1}{e}}}\br{\gamma\frac{1}{2}+\br{1-\gamma}h\br{\delta}},\gamma\cdot\delta\right\} .
\]
Optimizing over $\delta$ gives $\delta=0.525$ and a final approximation
factor strictly less than $4.07$. 

\section{Computing Dumping Factors}

\label{sec:computeDumping}

Recall that we assumed the knowledge of quantities $\beta_{ab}$,
which are needed to define the dumping factors $\alpha_{ab}$. Though
we are not able to compute the first quantities exactly in polynomial
time, we can efficiently estimate them and this is sufficient for
our goals. Let us focus on a given edge $ab$. Recall that 
\begin{multline*}
\beta_{ab}:=\excondls{\h A{}_{b}\setminus a,Y}{\prod_{a'\in A_{b}:Y_{a'b}<Y_{ab}}\br{1-p_{a'b}}}{a\in\h A_{b}}\\
\geq\frac{1}{1+p_{ab}}\br{1-\exp\br{-\br{1+p_{ab}}\frac{1}{p_{ab}}\ln\frac{1}{1-p_{ab}}}}=h\br{p_{ab}}.
\end{multline*}

Let us simulate the subroutine for buyers $N$ times without the dumping
factors ---- in a simulation we run GKPS, we sample the $Y$ variables,
but we simulate probes of edges, and we never really probe any edge.
We shall set $N$ later. Let $S^{1},S^{2},...,S^{N}$ be $0$-$1$
indicator random variables of whether $a$ was safe or not in each
simulation. Note that $\ex{S^{i}}=\beta_{ab}x_{ab}\in\brq{h\br{p_{ab}}x_{ab},x_{ab}}$.

Suppose that $x_{ab}\geq\frac{\epsilon}{n}$, where $n$ is the number
of buyers. The expression $\hat{s}_{ab}=\frac{1}{N}\sum_{i=1}^{N}S^{i}$
should be a good estimation of $\beta_{ab}\cdot x_{ab}$, i.e., $\hat{s}_{ab}\in\brq{\beta_{ab}x_{ab}\br{1-\epsilon},\beta_{ab}x_{ab}\br{1+\epsilon}}$
with probability $1-\frac{1}{n^{C}}$. \global\long\def\Z{Z}
 Set $N=\frac{6n}{\epsilon^{3}}\ln\br{2n^{2}\Z}$ for $Z=3\frac{1}{\eps}+1$.

Applying Chernoff's bound $\pr{|X-\ex X|>\eps\ex X}\leq2e^{-\frac{\epsilon^{2}}{3}\ex X}$
with $X=\sum_{i=1}^{N}S_{i}$ one obtains: 
\begin{align*}
 & \pr{\sum_{i=1}^{N}S_{i}\notin\brq{\br{1-\epsilon}\beta_{ab}x_{ab}\cdot N,\br{1+\epsilon}\beta_{ab}x_{ab}\cdot N}}\\
\leq & 2\exp\br{-\frac{\epsilon^{2}}{3}\beta_{ab}x_{ab}\cdot N}\leq2\exp\br{-\frac{\epsilon^{2}}{3}\frac{x_{ab}}{2}\cdot N}\leq2\exp\br{-\frac{\epsilon^{3}}{6n}\cdot N}=\frac{1}{n^{2}}\frac{1}{\Z}.
\end{align*}

From the union-bound, with probability at least $1-\frac{1}{\Z}$
we have that $\h s_{ab}\in\brq{\beta_{ab}x_{ab}\br{1-\epsilon},\beta_{ab}x_{ab}\br{1+\epsilon}}$
for every edge $ab$ such that $x_{ab}\geq\frac{\epsilon}{n}$.

Now let us assume this happened, i.e., we have good estimates. We
set $\alpha_{ab}=\max\{\frac{1}{2},\min\{\frac{1}{2}\frac{x_{ab}}{\h s_{ab}},1\}\}$
which belongs to $\brq{\frac{1}{2}\frac{1}{\beta_{ab}\br{1+\epsilon}},\frac{1}{2}\frac{1}{\beta_{ab}\br{1-\epsilon}}}$,
but only for edges $ab$ such that $x_{ab}\geq\frac{\epsilon}{n}$.
For edges $ab$ such that $x_{ab}<\frac{\epsilon}{n}$ we just put
$\alpha_{ab}=1$ (so we do not dump such edges actually). Two elements
of the proof were depending on the dumping factors. First, now the
probability that edge is taken is $\alpha_{ab}\beta_{ab}x_{ab}\in\brq{\frac{x_{ab}}{2\br{1+\epsilon}},\frac{x_{ab}}{2\br{1-\epsilon}}}$.
Second, recall the probability of an edge $ab$ not to be blocked:
\begin{equation}
\frac{1}{1+\br{1-\frac{1}{e}}\br{\sum_{b'\in B\setminus b}p_{ab'}\cdot\alpha_{ab'}\beta_{ab'}\cdot x_{ab'}}}.\label{eq:itemblockedwithdumping}
\end{equation}
We have that 
\begin{align*}
 & \sum_{b'\in B\setminus b}p_{ab'}\cdot\alpha_{ab'}\beta_{ab'}\cdot x_{ab'}\\
= & \sum_{b'\in B\setminus b:x_{ab'}\geq\frac{\epsilon}{n}}p_{ab'}\cdot\alpha_{ab'}\beta_{ab'}\cdot x_{ab'}+\sum_{b'\in B\setminus b:x_{ab'}<\frac{\epsilon}{n}}p_{ab'}\cdot\alpha_{ab'}\beta_{ab'}\cdot x_{ab'}\\
\leq & \sum_{b'\in B\setminus b:x_{ab'}\geq\frac{\epsilon}{n}}p_{ab'}\cdot\frac{1}{2\br{1-\epsilon}}x_{ab'}+\sum_{b'\in B\setminus b:x_{ab'}<\frac{\epsilon}{n}}x_{ab'}\\
\leq & \frac{1}{2\br{1-\epsilon}}+\epsilon=\frac{1}{2}+O\br{\epsilon}.
\end{align*}
So the probability that $a$ is not blocked is at least $\frac{1}{1+\br{1-\frac{1}{e}}\br{\frac{1}{2}+O\br{\epsilon}}}.$
The final probability that edge $ab$ is probed is at least 
\begin{align*}
\br{1-\frac{1}{e}}\frac{x_{ab}}{2\br{1+\epsilon}}\cdot\frac{1}{1+\br{1-\frac{1}{e}}\br{\frac{1}{2}+O\br{\epsilon}}} & =\frac{x_{ab}}{1+\eps}\cdot\frac{e-1}{2e+\br{e-1}\br{1+O\br{\epsilon}}}\\
 & =x_{ab}\cdot\frac{e-1}{3e-1+O\br{\epsilon}}>\frac{1}{4.16}\cdot x_{ab}.
\end{align*}
In the last inequality above we assumed $\eps$ to be small enough.

With probability at most $\frac{1}{\Z}$ we did not obtain good estimates
of the dumping factors. Still we have that $\alpha_{ab}\in\brq{\frac{1}{2},1}$,
and therefore $\alpha_{ab}\beta_{ab}\in\brq{\frac{1}{4},1}$. In this
case quantity~\eqref{eq:itemblockedwithdumping} can be just lower-bounded
by $\frac{1}{1+\br{1-\frac{1}{e}}}$, and the probability that edge
$ab$ is probed in the subroutine for buyers is at least $\frac{x_{ab}}{4}$.
Thus the probability that edge $ab$ is probed during the
algorithm is at least $\br{1-\frac{1}{e}}\frac{x_{ab}}{4}\cdot\frac{1}{1+\br{1-\frac{1}{e}}}=\frac{x_{ab}}{4}\cdot\frac{e-1}{2e-1}>\frac{1}{10.33}x_{ab}.$
The total expected outcome of the algorithm is therefore, for sufficiently
small $\eps$, at least 
\[
LP_{onl}\br{\br{1-\frac{1}{\Z}}\frac{e-1}{3e-1+O\br{\epsilon}}+\frac{1}{\Z}\frac{1}{4}\cdot\frac{e-1}{2e-1}}\overset{Z=3\frac{1}{\eps}+1}{\geq}\frac{1}{4.16}LP_{onl}.
\]


The above approach can be combined with the small/big probability
trick from Appendix \ref{sec:onlineBigSmall}. 
By choosing $\eps$ small enough the approximation ratio is $4.07$
as claimed.

\end{document}